\documentclass[runningheads]{llncs}
\usepackage{array}

\usepackage{graphicx}
\begin{document}
\title{A Further Improvement on Approximating TTP-2}
%
%
\author{Jingyang Zhao \and Mingyu Xiao\orcidID{0000-0002-1012-2373}}
\authorrunning{J. Zhao and M. Xiao}
%
\institute{University of Electronic Science and Technology of China, Chengdu, China\\
\email{1176033045@qq.com, myxiao@gmail.com}}
\maketitle              
\begin{abstract}
The Traveling Tournament Problem (TTP) is a hard but interesting sports scheduling problem inspired by Major League Baseball, which is to design a double round-robin schedule such that each pair of teams plays one game in each other's home venue, minimizing the total distance traveled by all $n$ teams ($n$ is even). In this paper, we consider TTP-2, i.e., TTP with one more constraint that each team can have at most two consecutive home games or away games. Due to the different structural properties, known algorithms for TTP-2 are different for $n/2$ being odd and even.
For odd $n/2$, the best known approximation ratio is about $(1+12/n)$, and for even $n/2$, the best known approximation ratio is about $(1+4/n)$.
In this paper, we further improve the approximation ratio from $(1+4/n)$ to $(1+3/n)$ for $n/2$ being even.
Experimental results on benchmark sets show that our algorithm can improve previous results on all instances with even $n/2$ by $1\%$ to $4\%$.

\end{abstract}
\section{Introduction}
The Traveling Tournament Problem (TTP), first systematically introduced in~\cite{easton2001traveling}, is a hard but interesting sports scheduling problem inspired by Major League Baseball.
This problem is to find a double round-robin tournament satisfying several constraints that minimizes the total distances traveled by all participant teams.
There are $n$ participating teams in the tournament, where $n$ is always even. Each team should play $2(n-1)$ games in $2(n-1)$ consecutive days. Since each team can only play one game on each day, there are exact $n/2$ games scheduled on each day.
There are exact two games between any pair of teams,
where one game is held at the home venue of one team and the other one is held at the home venue of the other team.
The two games between the same pair of teams could not be scheduled in two consecutive days.
These are the constraints for TTP. We can see that it is not easy to construct a feasible schedule.
Now we need to find an optimal schedule that minimizes the total traveling distances by all the $n$ teams.
A well-known variant of TTP is TTP-$k$. which has one more constraint:
each team is allowed to take at most $k$ consecutive home or away games.
If $k$ is very large, say $k=n-1$, then this constraint will lose its meaning and it becomes TTP again. For this case, a team can schedule its travel distance as short as the traveling salesmen problem. On the other hand,
in a sports schedule, it is generally believed that home stands and road trips should alternate as regularly as possible for each team~\cite{campbell1976minimum,thielen2012approximation}.
The smaller the value of $k$, the more frequently teams have to return their homes.
TTP and its variants have been extensively studied in the literature~\cite{kendall2010scheduling,rasmussen2008round,thielen2012approximation,DBLP:conf/mfcs/XiaoK16}.

\subsection{Related Work}
In this paper, we will focus on TTP-2. We mainly survey the results on TTP-$k$.
For $k=1$, TTP-1 is trivial and there is no feasible schedule~\cite{de1988some}.
But when $k\geq 2$, the problem suddenly becomes very hard. It is not easy to find a simple feasible schedule. Even no good
brute force algorithm with a single exponential running time has been found yet.
In the online benchmark \cite{trick2007challenge}, most instances with more than $10$ teams are still unsolved completely even by using high-performance machines.
The NP-hardness of
TTP-$k$ with $k=3$ or $k=n-1$ has been proved \cite{bhattacharyya2016complexity,thielen2011complexity}.
Although the hardness of other cases has not been theoretically proved, most people believe TTP-$k$ with $k\geq 2$ is very hard.
In the literature, there is a large number of contributions on approximation algorithms~\cite{yamaguchi2009improved,imahori2010approximation,miyashiro2012approximation,westphal2014,hoshino2013approximation,thielen2012approximation,DBLP:conf/mfcs/XiaoK16} and heuristic algorithms~\cite{easton2003solving,lim2006simulated,anagnostopoulos2006simulated,di2007composite,goerigk2014solving}.

In terms of approximation algorithms, most results are based on the assumption that the distance holds the symmetry and triangle inequality properties. This is natural and practical in the sports schedule.
For TTP or TTP-$k$ with $k\geq n-1$, Westphal and Noparlik \cite{westphal2014} proved an approximation ratio of 5.875 and Imahori\emph{ et al.} \cite{imahori2010approximation} proved an approximation ratio of 2.75 at the same time.
For TTP-3, the current approximation ratio is $5/3+O(1/n)$~\cite{yamaguchi2009improved}.
The first record of TTP-2 seems from the schedule of a basketball conference of ten teams
in~\cite{campbell1976minimum}. This paper did not discuss the approximation ratio.
In fact, any feasible schedule for TTP-2
is a 2-approximation solution~\cite{thielen2012approximation}.
Although any feasible schedule will not have a very bad performance, no simple construction of feasible schedules is known now.
In the literature, all known algorithms for TTP-2 are different for $n/2$ being even and odd. This may be caused by different structural properties. One significant contribution to TTP-2 was done by Thielen and Westphal~\cite{thielen2012approximation}.
They proposed a $(3/2+O(1/n))$-approximation algorithm for $n/2$ being odd and a $(1+16/n)$-approximation algorithm for $n/2$ being even.
Now the approximation ratio was improved to $(1+\frac{12}{n}+\frac{8}{n(n-2)})$ for odd $n/2$~\cite{DBLP:conf/ijcai/ZhaoX21} and to $(1+\frac{4}{n}+\frac{4}{n(n-2)})$ for even $n/2$~\cite{DBLP:conf/mfcs/XiaoK16}.

\subsection{Our Results}
In this paper, we design an effective algorithm for TTP-2 with $n/2$ being even with an approximation ratio $(1+\frac{3}{n}-\frac{6}{n(n-2)})$, improving the ratio from $(1+\frac{4}{n}+\Theta(\frac{1}{n(n-2)}))$ to
$(1+\frac{3}{n}-\Theta(\frac{1}{n(n-2)}))$. Now the ratio is small and improvement becomes harder and harder.
Our major algorithm is based on packing minimum perfect matching. We first find a minimum perfect matching in the distance graph, then pair the teams according to the matching, and finally construct a feasible schedule based on the paired teams (called super-teams).
Our algorithm is also easy to implement and runs fast.
Experiments show that our results beat all previously-known solutions on the 17 tested instances in~\cite{DBLP:conf/mfcs/XiaoK16} with an average improvement of $2.10\%$.

\section{Preliminaries}\label{sec_pre}
We will always use $n$ to denote the number of teams and let $m=n/2$, where $n$ is an even number.
We also use $\{t_1, t_2, \dots, t_n\}$ to denote the set of the $n$ teams.
A sports scheduling on $n$ teams is \emph{feasible} if it holds the following properties.
\begin{itemize}
\item \emph{Fixed-game-value}: Each team plays two games with each of the other $n-1$ teams, one at its home venue and one at its opponent's home venue.
\item \emph{Fixed-game-time}: All the games are scheduled in $2(n-1)$ consecutive days and each team plays exactly one game in each of the $2(n-1)$ days. 
\item \emph{Direct-traveling}: All teams are initially at home before any game begins, all teams will come back home after all games, and a team travels directly from its game venue in the $i$th day to its game venue in the $(i+1)$th day.
\item \emph{No-repeat}: No two teams play against each other on two consecutive days.
\item \emph{Bounded-by-$k$}: The number of consecutive home/away games for any team is at most $k$.
\end{itemize}

The TTP-$k$ problem is to find a feasible schedule minimizing the total traveling distance of all the $n$ teams.
The input of TTP-$k$ contains an $n \times n$ distance matrix $D$ that indicates the distance between each pair of teams.
The distance from the home of team $i$ to the home of team $j$ is denoted by $D_{i,j}$.
We also assume that $D$ satisfies the symmetry and triangle inequality properties, i.e., $D_{i,j}=D_{j,i}$ and $D_{i,j} \leq D_{i,h} + D_{h,j}$ for all $i,j,h$. We also let $D_{i,i}=0$ for each $i$.

We will use $G$ to denote an edge-weighted complete graph on $n$ vertices representing the $n$ teams.
The weight of the edge between two vertices $t_i$ and $t_j$ is $D_{i,j}$, the distance from the home of $t_i$ to the home of $t_j$.
We also use $D_i$ to denote the weight sum of all edges incident on $t_i$ in $G$, i.e., $D_i=\sum_{j=1}^n D_{i,j}$.
The sum of all edge weights of $G$ is denoted by $D_G$.

We let $M$ denote a minimum weight perfect matching in $G$. The weight sum of all edges in $M$ is denoted by $D_M$.
We may consider the endpoint pair of each edge in $M$ as a \emph{super-team}. We use $H$ to denote the complete graph on the $m$ vertices representing the $m$ super-teams. The weight of the edge between two super-teams $u_i$ and $u_j$, denoted by $D(u_i,u_j)$, is the sum of the weight of the four edges in $G$ between one team in $u_i$ and one team in $u_j$, i.e., $D(u_i, u_j)=\sum_{t_{i'}\in u_i \& t_{j'}\in u_j}D_{i',j'}$. We also let $D(u_i,u_i)=0$ for any $i$.
We give an illustration of the graphs $G$ and $H$ in Figure~\ref{fig:fig001}.
\begin{figure}[ht]
    \centering
    \includegraphics[scale=0.6]{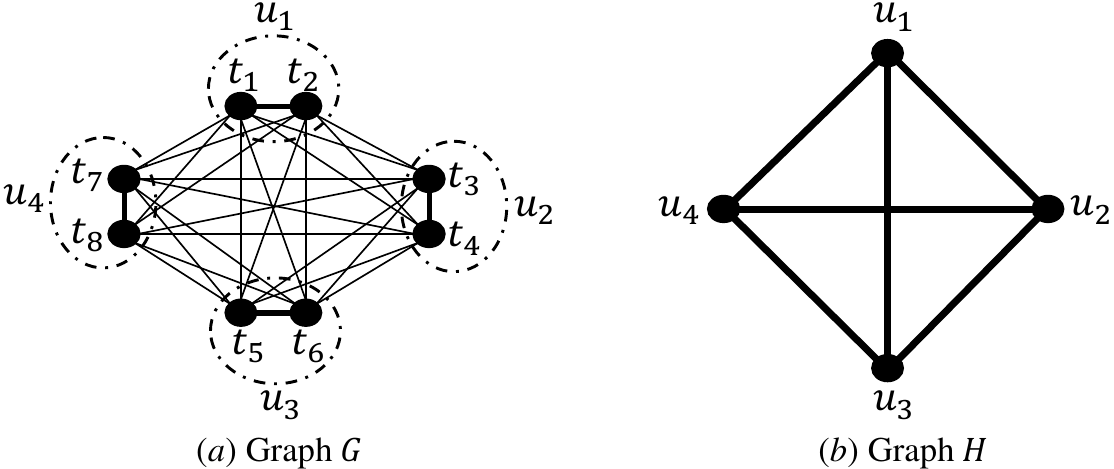}
    \caption{An illustration of graphs $G$ and $H$, where there four dark lines form a minimum perfect matching $M$ in $G$}
    \label{fig:fig001}
\end{figure}

The sum of all edge weights of $H$ is denoted by $D_H$. It holds that
\begin{eqnarray} \label{eqn_GH}
D_H=D_G-D_M.
\end{eqnarray}

\subsection{Independent lower bound and extra cost}
The \emph{independent lower bound} for TTP-2 was firstly introduced by Campbell and Chen~\cite{campbell1976minimum}.
It has become a frequently used lower bound.
The basic idea of the independent lower bound is to obtain a lower bound $LB_i$ on the traveling distance of a single team $t_i$ independently without considering the feasibility of other teams.

The road of a team $t_i$ in TTP-$2$, starting at its home venue and coming back home after all games, is called
an \emph{itinerary} of the team. The itinerary of $t_i$ is also regarded as a graph on the $n$ teams,
which is called the \emph{itinerary graph} of $t_i$.
In an itinerary graph of $t_i$, the degree of all vertices except $t_i$ is 2 and the degree of $t_i$ is greater than or equal to $n$ since team $t_i$ will visit each other team venue only once.
Furthermore, for any other team $t_j$, there is at least one edge between $t_i$ and $t_j$, because $t_i$ can only visit at most 2 teams on each road trip and then team $t_i$ either comes from its home to team $t_j$ or goes back to its home after visiting team $t_j$. We decompose the itinerary graph of $t_i$ into two parts: one is a spanning star centered at $t_i$ (a spanning tree which only vertex $t_i$ of degree $> 1$) and the forest of the remaining part. Note that in the forest, only $t_i$ may be a vertex of degree $\geq 2$ and all other vertices are degree-1 vertices. See Figure~\ref{fig:fig002} for illustrations of the itinerary graphs.

\begin{figure}[ht]
    \centering
    \includegraphics[scale=0.85]{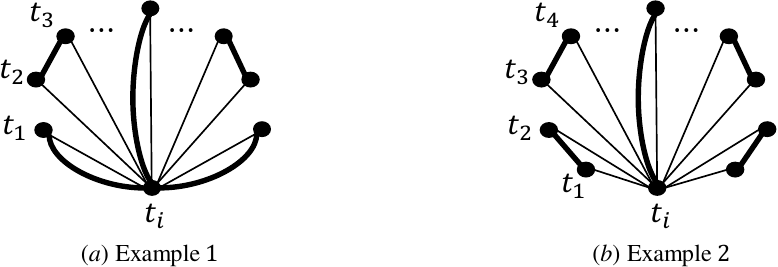}
    \caption{The itinerary graph of $t_i$, where the light edges form a spanning star and the dark edges form the remaining forest. In the right example (b), the remaining forest is a perfect matching of $G$}
    \label{fig:fig002}
 \end{figure}

For different itineraries of $t_i$, the spanning star is fixed and only the remaining forest may be different.
The total distance of the spanning star is $\sum_{j\neq i} D_{i,j}=D_i$. Next, we show an upper and lower bound on the total distance of the remaining forest. For each edge between two vertices $t_{j_1}$ and $t_{j_2}$ ($j_1,j_2\neq i$), we have that $D_{j_1,j_2}\leq D_{i,j_1}+D_{i,j_2}$ by the triangle inequality property. Thus, we know that the total distance of the remaining forest is at most the total distance of the spanning star. Therefore, the distance of any feasible itinerary of $t_i$ is at most $2D_i$. This also implies that any feasible solution to TTP-2 is a 2-approximation solution.
 On the other hand, the distance of the remaining forest is at least as that of a minimum perfect matching of $G$ by the triangle inequality.
Recall that we use $M$ to denote a minimum perfect matching of $G$. Thus, we have
a lower bound $LB_i$ for each team $t_i$:
\begin{eqnarray} \label{eqn_lower1}
LB_i=D_i+D_M.
\end{eqnarray}

The itinerary of $t_i$ to achieve $LB_i$ is called the \emph{optimal itinerary}.
The \emph{independent lower bound} for TTP-2 is the traveling distance such that all teams reach their optimal itineraries, which is denoted as
\begin{eqnarray} \label{eqn_lowerbound}
 LB=\sum_{i=1}^n LB_i =\sum_{i=1}^n (D_i +D_M)=2D_G+nD_M.
\end{eqnarray}

For any team, it is possible to reach its optimal itinerary. However, it is impossible for all teams to reach their optimal itineraries synchronously in a feasible schedule~\cite{thielen2012approximation}, even for $n=4$. So the independent lower bound for TTP-2 is not achievable.

To analyze the quality of a schedule of the tournament, we will compare the itinerary of each team with the optimal itinerary.
The different distance is called the \emph{extra cost}. Sometimes it is not convenient to compare the whole itinerary directly.
We may consider the extra cost for a subpart of the itinerary.
We may split an itinerary into several trips and each time we compare some trips.
A \emph{road trip} in an itinerary of team $t_i$ is a simple cycle starting and ending at $t_i$.
So an itinerary consists of several road trips. For TTP-2, each road trip is a triangle or a cycle on two vertices.
Let $L$ and $L'$ be two itineraries of team $t_i$, $L_s$ be a sub itinerary of $L$ consisting of several road trips in $L$, and
$L'_s$ be a sub itinerary of $L'$ consisting of several road trips in $L'$.
We say that the sub itineraries $L_s$ and $L'_s$ are \emph{coincident} if they visit the same set of teams.
We will only compare a sub itinerary of our schedule with a coincident sub itinerary of the optimal itinerary and consider the extra cost between them.

\section{Constructing the Schedule}
We will introduce a method to construct a feasible tournament first.
Our construction consists of two parts. First, we arrange \emph{super-games} between \emph{super-teams}, where each super-team contains
a pair of normal teams. Then we extend super-games to normal games between normal teams.
To make the itinerary as similar as the optimal itinerary, we take each team pair in the minimum perfect matching $M$ of $G$ as a \emph{super-team}. There are $n$ normal teams and then there are $m=n/2$ super-teams. We denote the set of super-teams as $\{u_1, u_2, \dots, u_{m}\}$ and relabel the $n$ teams such that $u_i=\{t_{2i-1},t_{2i}\}$ for each $i$.

Each super-team will attend $m-1$ super-games in $m-1$ time slots.
Each super-game on the first $m-2$ time slots will be extended to eight normal games between normal teams on four days, and each super-game on the last time slot will be extended to twelve normal games between normal teams on six days. So each normal team $t_i$ will attend $4\times (m-2)+6=4m-2=2n-2$ games.
This is the number of games each team $t_i$ should attend in TTP-2.
In our algorithm, the case of $n=4$ is easy, and hence we assume here that $n\geq 8$.

We construct the schedule for super-teams from the first time slot to the last time slot $m-1$.
In each of the $m-1$ time slots, we have $\frac{m}{2}$ super-games.
In fact, our schedules in the first time slot and in the last time slot are different from the schedules in the middle time slots.

For the first time slot, the $\frac{m}{2}$ super-games are arranged as shown in Figure~\ref{fig:figa}. All of these super-games are called \emph{normal super-games}. Each super-game is represented by a directed edge, the information of which will be used to extend super-games to normal games between normal teams.

\begin{figure}[ht]
    \centering
    \includegraphics[scale=0.45]{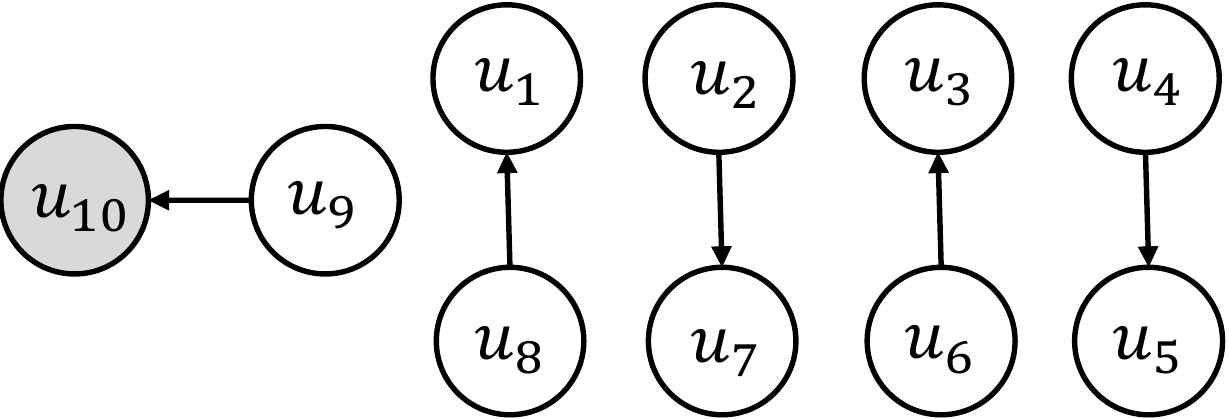}
    \caption{The super-game schedule on the first time slot for an instance with $m=10$}
    \label{fig:figa}
\end{figure}

In Figure~\ref{fig:figa}, the last super-team $u_m$ is denoted as a dark node, and all other super-teams
$u_1, \dots, u_{m-1}$ are denoted as white nodes.
The while nodes form a cycle and we may change the positions of the while nodes according to the cycle in the other time slots.
In the second time slot, we keep the position of $u_m$ and change the positions of white super-teams in the cycle by moving one position in the clockwise direction, and also change the direction of each edge except for the most left edge incident on $u_m$. This edge will be replaced by a double arrow edge. The super-game including $u_m$ is also called a \emph{left super-game} in the middle $m-3$ time slots. So in the second time slot, there are $\frac{m}{2}-1$ normal super-games and one left super-games.
 An illustration of the schedule in the second time slot is shown in Figure~\ref{fig:figb}.

 \begin{figure}[ht]
    \centering
    \includegraphics[scale=0.45]{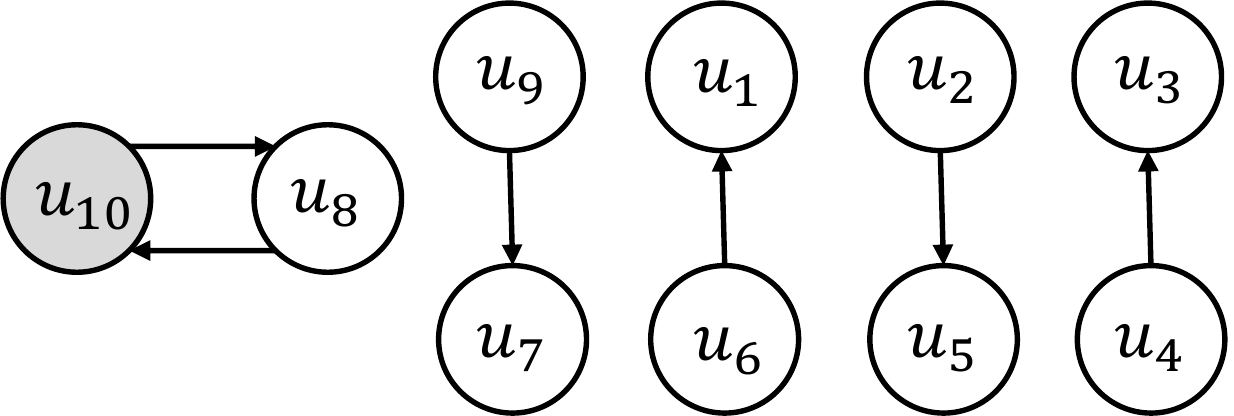}
    \caption{The super-game schedule on the second time slot for an instance with $m=10$}
    \label{fig:figb}
\end{figure}

In the third time slot, there are also $\frac{m}{2}-1$ normal super-games and one left super-games.
We also change the positions of white super-teams in the cycle by moving one position in the clockwise direction while the direction of each edge is reversed. The position of the dark node will always keep the same.
An illustration of the schedule in the third time slot is shown in Figure~\ref{fig:figc}.

\begin{figure}[ht]
    \centering
    \includegraphics[scale=0.45]{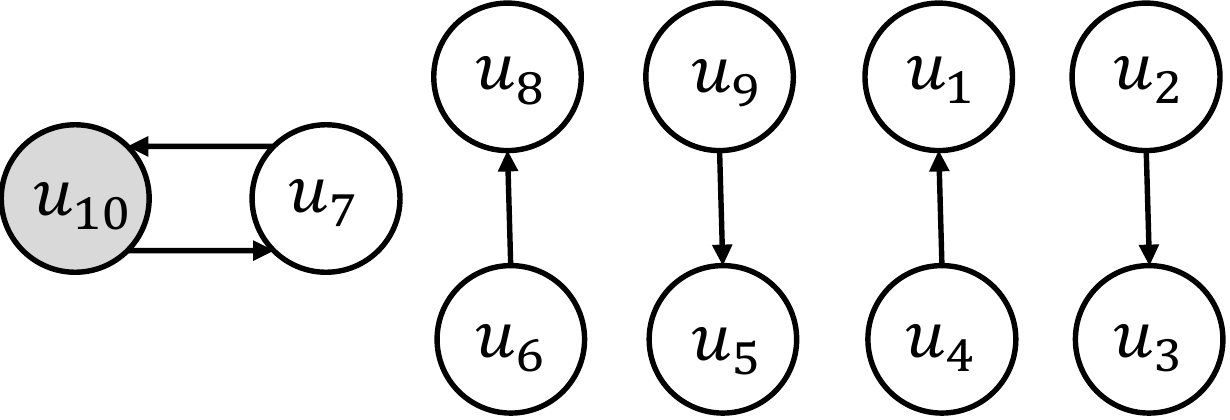}
    \caption{The super-game schedule on the third time slot for an instance with $m=10$}
    \label{fig:figc}
\end{figure}

The schedules for the other middle slots are derived analogously. Before we introduce the super-games in the last time slot $m-1$,
 we first explain how to extend the super-games in the first $m-2$ time slots to normal games.
 In these time slots, we have two different kinds of super-games: normal super-games and left super-games. We first consider normal super-games.

\textbf{Case~1. Normal super-games}:
Each normal super-game will be extended to eight normal games on four days.
Assume that in a normal super-game, super-team $u_{i}$ plays against the super-team $u_{j}$ on time slot $q$ ($1\leq i,j\leq m$ and $1\leq q\leq m-2$). Recall that $u_{i}$ represents normal teams \{$t_{2i-1}, t_{2i}$\} and $u_{j}$ represents normal teams \{$t_{2j-1}, t_{2j}$\}. The super-game will be extended to eight normal games on four corresponding days from $4q-3$ to $4q$, as shown in Figure~\ref{fig:fig003}. A directed edge from team $t_{i'}$ to team $t_{i''}$ means $t_{i'}$ plays against $t_{i''}$ at the home venue of $t_{i''}$.
Note that if there is a directed edge from $u_j$ to $u_i$, then the direction of all the edges in Figure~\ref{fig:fig003} should be reversed.

\begin{figure}[ht]
    \centering
    \includegraphics[scale=0.75]{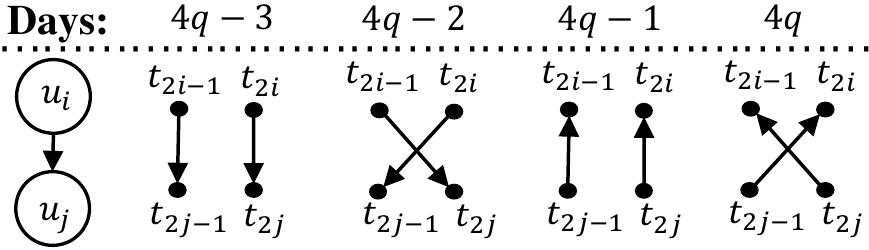}
    \caption{Extending normal super-games}
    \label{fig:fig003}
\end{figure}

\textbf{Case~2. Left super-games}:
Assume that in a left super-game, super-team $u_{m}$ plays against super-team $u_{i}$ on time slot $q$ ($2\leq i\leq m-2$ and $2\leq q\leq m-2$). Recall that $u_{m}$ represents normal teams \{$t_{2m-1}, t_{2m}$\} and $u_{i}$ represents normal teams \{$t_{2i-1}, t_{2i}$\}. The super-game will be extended to eight normal games on four corresponding days from $4q-3$ to $4q$, as shown in Figure~\ref{fig:fig004} for even time slot $q$. For odd time slot $q$, the direction of edges in the figure will be reversed.

\begin{figure}[ht]
    \centering
    \includegraphics[scale=0.8]{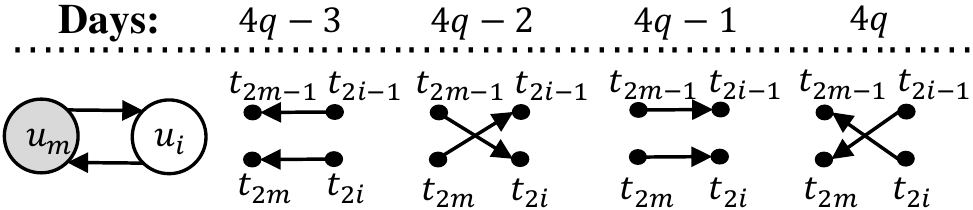}
    \caption{Extending left super-games}
    \label{fig:fig004}
\end{figure}

The first $m-2$ time slots will be extended to $4(m-2)=2n-8$ days according to the above rules. Each normal team will have six remaining games,
which will be corresponding to the super-games on the last time slot.
We will call a super-game on the last time slot a \emph{last super-game}.
Figure~\ref{fig:fig005} shows an example of the schedule on the last time slot.

\begin{figure}[ht]
    \centering
    \includegraphics[scale=0.45]{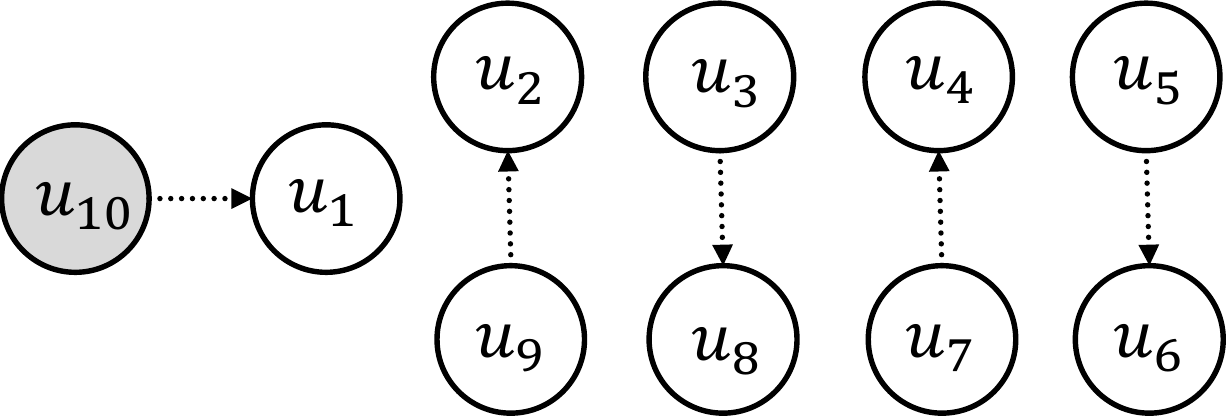}
    \caption{The super-game schedule on the last time slot for an instance with $m=10$}
    \label{fig:fig005}
\end{figure}

\textbf{Case~3. Last super-games}:
Next, we extend a last super-game into twelve normal games on six days.
Assume that on the last time slot $q=m-1$, super-team $u_i$ plays against super-team $u_j$ ($1\leq i,j\leq m$). Recall that $u_{i}$ represents normal teams \{$t_{2i-1}, t_{2i}$\} and $u_{j}$ represents normal teams \{$t_{2j-1}, t_{2j}$\}. The last super-game will be extended to twelve normal games on six corresponding days from $4q-3$ to $4q+2$, as shown in Figure~\ref{fig:fig006}.

\begin{figure}[ht]
    \centering
    \includegraphics[scale=0.75]{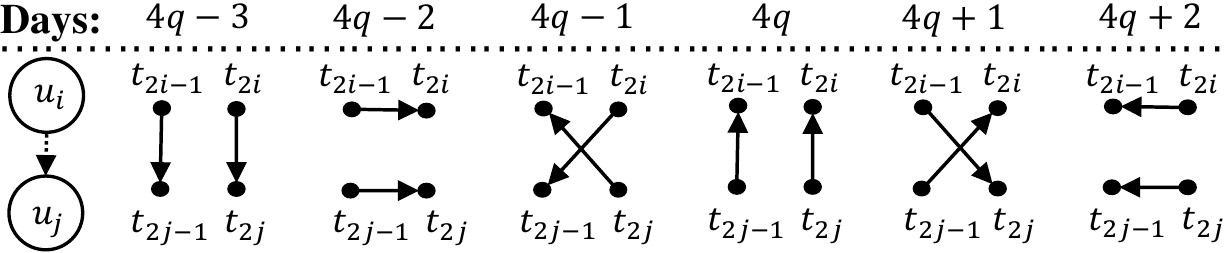}
    \caption{Extending last super-games}
    \label{fig:fig006}
\end{figure}

The above is the main part of the schedule. Now, we give an example of the schedule for $n=8$ teams constructed by using the above rules. In Table~\ref{ansexample}, the $i$th row indicates team $t_i$, the $j$th column indicates the $j$th day in the double round-robin,
 item $+t_{x}$ on the $i$th row and $j$th column means that team $t_i$ plays against team $t_{x}$ on the $j$th day at the home venue of the opposite team $t_{x}$, and item $-t_{x}$ on the $i$th row and $j$th column means that team $t_i$ plays against team $t_{x}$ on the $j$th day at its home venue.

\begin{table}[ht]
\centering
\begin{tabular}{ m{0.3cm}<{\centering}|*{14}{m{0.70cm}<{\centering}} }
   & 1 & 2 & 3 & 4 & 5 & 6 & 7 & 8 & 9 & 10 & 11 & 12 & 13 & 14\\
   \hline
  $t_{1}$ & $-t_{3}$ & $-t_{4}$ & $+t_{3}$ & $+t_{4}$ & $-t_{5}$ & $-t_{6}$ & $+t_{5}$ & $+t_{6}$ &
            $-t_{7}$ & $+t_{2}$ & $-t_{8}$ & $+t_{7}$ & $+t_{8}$ & $-t_{2}$ \\
  $t_{2}$ & $-t_{4}$ & $-t_{3}$ & $+t_{4}$ & $+t_{3}$ & $-t_{6}$ & $-t_{5}$ & $+t_{6}$ & $+t_{5}$ &
            $-t_{8}$ & $-t_{1}$ & $+t_{7}$ & $+t_{8}$ & $-t_{7}$ & $+t_{1}$ \\
  $t_{3}$ & $+t_{1}$ & $+t_{2}$ & $-t_{1}$ & $-t_{2}$ & $+t_{7}$ & $-t_{8}$ & $-t_{7}$ & $+t_{8}$ &
            $-t_{5}$ & $+t_{4}$ & $-t_{6}$ & $+t_{5}$ & $+t_{6}$ & $-t_{4}$ \\
  $t_{4}$ & $+t_{2}$ & $+t_{1}$ & $-t_{2}$ & $-t_{1}$ & $+t_{8}$ & $-t_{7}$ & $-t_{8}$ & $+t_{7}$ &
            $-t_{6}$ & $-t_{3}$ & $+t_{5}$ & $+t_{6}$ & $-t_{5}$ & $+t_{3}$ \\
  $t_{5}$ & $+t_{7}$ & $+t_{8}$ & $-t_{7}$ & $-t_{8}$ & $+t_{1}$ & $+t_{2}$ & $-t_{1}$ & $-t_{2}$ &
            $+t_{3}$ & $+t_{6}$ & $-t_{4}$ & $-t_{3}$ & $+t_{4}$ & $-t_{6}$ \\
  $t_{6}$ & $+t_{8}$ & $+t_{7}$ & $-t_{8}$ & $-t_{7}$ & $+t_{2}$ & $+t_{1}$ & $-t_{2}$ & $-t_{1}$ &
            $+t_{4}$ & $-t_{5}$ & $+t_{3}$ & $-t_{4}$ & $-t_{3}$ & $+t_{5}$ \\
  $t_{7}$ & $-t_{5}$ & $-t_{6}$ & $+t_{5}$ & $+t_{6}$ & $-t_{3}$ & $+t_{4}$ & $+t_{3}$ & $-t_{4}$ &
            $+t_{1}$ & $+t_{8}$ & $-t_{2}$ & $-t_{1}$ & $+t_{2}$ & $-t_{8}$ \\
  $t_{8}$ & $-t_{6}$ & $-t_{5}$ & $+t_{6}$ & $+t_{5}$ & $-t_{4}$ & $+t_{3}$ & $+t_{4}$ & $-t_{3}$ &
            $+t_{2}$ & $-t_{7}$ & $+t_{1}$ & $-t_{2}$ & $-t_{1}$ & $+t_{7}$ \\
\end{tabular}

\caption{The schedule for $n=8$ teams, where the horizontal ordinate represents the teams, the ordinate represents the days,
`$+$' means the team on the corresponding horizontal ordinate plays at its home, and  `$-$' means the team on the corresponding horizontal ordinate plays at the opposite team's home}
\label{ansexample}
\end{table}

From Table~\ref{ansexample}, we can see that on each line there are at most two consecutive `$+$'/`$-$', and then we can see that this is a feasible schedule.

\begin{theorem}
For TTP-$2$ with $n$ teams such that $n\equiv 0 \pmod 4$, the above construction can generate a feasible schedule.
\end{theorem}

\begin{proof}
According to the definition of feasible schedules, we only need to prove the five properties in the definition.

The first two properties of fixed-game-value and fixed-game-time are easy to see from the construction.
Each super-game on the first $m-2$ time slot will be extended to eight normal games on four days and each team participates in four games on four days. Each super-game on the last time slot will be extended to twelve normal games on six days and each team participates in six games on six days. So each team plays $2(n-1)$ games on $2(n-1)$ different days. It is also easy to see from the construction that each team pair plays exactly two games, one at the home venue of each team.
 We also assume that the itinerary obeys the direct-traveling property. It does not need to prove.

Next, we focus on the bounded-by-$2$ property. We will use `$H$' and `$A$' to denote a home game and an away game, respectively. We will also let $\overline{H}=A$ and $\overline{A}=H$.

We first look at the games in the first $2n-8$ days. For the two teams in $u_m$, the 4 games in the first time slot will be $HHAA$, in an even time slot will be $HAAH$ (see Figure~7, and the 4 games in an odd time slot (not including the first time slot) will be $AHHA$. So two consecutive time slots can be jointed well.
Next, we consider a team $t_i$ in $u_j$ $(j\in \{1,2,\dots,m-1\})$.
In the time slots for normal super-games, the 4 games for $t_i$ will be $AAHH$ if the arc between the two super-teams is out of $u_j$ and $\overline{AAHH}$ otherwise. In the time slots for left super-games, the 4 games will be $AHHA$ or $\overline{AHHA}$.
In the time slot after a left super-game, the 4 games will be $HHAA$ or $\overline{HHAA}$.
In the time slot before a left super-game, the 4 games will be $AAHH$ or $\overline{AAHH}$.
So two consecutive time slots can joint well, no matter they are two time slots for normal super-games, or one time slot for a normal super-game and one time slot for a left super-game.

Only the last 6 days on the last time slot have not been analyzed yet. For the sake of presentation,
we simply list out all the games on the last two time slots for each team.
There are 10 games for each team, 4 on the second to last time slot and 6 on the last time slot.
Let $A\in \{3,\dots, m-1\}$ and $H\in \{4,6,\dots, m-2\}\cup\{1\}$ such that super-team $u_A$ plays $AAHH$ and $u_H$ plays $HHAA$ on the penultimate time slot. Note that the forms of super-teams $u_2$ and $u_m$ are different because they play against each other in a left super-game on the penultimate time slot.
We also denote teams $\{t_{2i-1},t_{2i}\}$ in super-team $u_i$ by $\{t_{i_1},t_{i_2}\}$.
The last 10 games for all the teams are shown in Figure~10.
We can see that there are no three consecutive home games or away games. So the \emph{bounded-by-k} property holds.

All the properties are satisfied and then the schedule is a feasible schedule for TTP-2.
\end{proof}

\begin{figure}[h]
    \centering
    \includegraphics[scale=0.7]{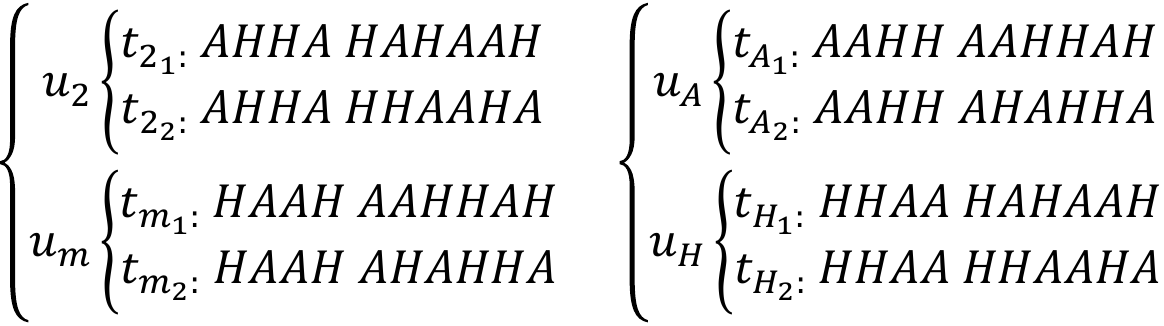}
    \caption{The last 10 games for the case of $n\equiv 0 \pmod 4$}
    \label{fig:fig007}
\end{figure}

We have introduced a method to construct a feasible schedule. However, it is not our final schedule. We may still need to specify the order of some teams or super-teams to minimize the extra cost. We will introduce this when we analyze the approximation ratio.

\section{Approximation Quality of the Schedule}
\subsection{Analysis of the approximation ratio}
To show the quality of our schedule, we compare it with the independent lower bound. We will check the difference between our itinerary of each team $t_i$ and the optimal itinerary of $t_i$ and compute the extra cost.
As mentioned in the last paragraph of Section~\ref{sec_pre}, we will compare some sub itineraries of a team.
We will look at the sub itinerary of a team on the four or six days in a super-game, which is coincident with a sub itinerary of the optimal itinerary:
all teams stay at home before the first game in a super-game and return home after the last game in the super-game.
In our algorithm, there are three types of super-games: normal super-games, left super-games, and last super-games. We analyze the total extra cost of all normal teams caused by each type of super-games.

\begin{lemma}\label{extra}
Assume there is a super-game between super-teams $u_i$ and $u_j$ in our schedule.
\begin{enumerate}
\item [(a)] If the super-game is a normal super-game, then the extra cost of all normal teams in $u_i$ and $u_j$ is 0;
\item [(b)] If the super-game is a left or last super-game, then the extra cost of all normal teams in $u_i$ and $u_j$ is at most $D(u_i,u_j)$.
\end{enumerate}
\end{lemma}
\begin{proof}
From Figure 6 we can see that in a normal super-game any of the four normal teams will visit the home venue of the two normal teams in the opposite super-team in one road trip. So they have the same road trip as that in their optimal itineraries.
The extra cost is 0. So (a) holds.

Next, we assume that the super-game is a left super-game and $u_i=u_m$. From Figure~\ref{fig:fig004}, we can see that the two teams $t_{n-1}$ and $t_{n}$ in the super-team $u_m$ play $AHHA$ in the four days (Recall that $A$ means an away game and $H$ means a home game), and the two teams 
in the super-team $u_j$ play $HAAH$.
The two teams in $u_j$ will have the same road trip as that in the optimal itinerary and then the extra cost is 0.
We compare the road trips of the two teams $t_{n-1}$ and $t_{n}$ with their optimal road trips together. The difference is that
$$D_{n-1,2j-1}+D_{n-1,2j}+D_{n,2j-1}+D_{n,2j}-2D_{2j-1,2j}\leq D(u_m,u_j),$$
 by the triangle inequality.

Last, we consider it as a last super-game. We assume with loss of generality that on the first day of the six days, the games are held at the home venue of teams in $u_i$.
From Figure ~\ref{fig:fig006}, we can see that teams in $u_j$ do not have any extra cost and teams in $u_i$ have a total extra cost of
 $$2D_{2i-1,2j}+D_{2i,2j-1}+D_{2i,2j}-D_{2i-1,2i}-2D_{2j-1,2j}\leq D(u_i,u_j),$$
 by the triangle inequality.
\end{proof}

In our schedule, there are $\frac{m}{2}+(m-3)(\frac{m}{2}-1)$ normal super-games, which contribute 0  to the extra cost.
There are $m-3$ left super-games on the $m-3$ middle time slots. By Lemma~\ref{extra}, we know that the total extra cost is $E_1=\sum_{i=2}^{m-2}D(u_m,u_i)$.
There are $\frac{m}{2}$ last super-games on the last time slot. By Lemma~\ref{extra}, we know that the total extra cost is  $E_2=\sum_{i=1}^{m/2}D(u_i,u_{m+1-i})$.
\begin{lemma}
The total extra cost of our schedule is at most
\[ E_1+E_2= \sum_{i=2}^{m-2}D(u_m,u_i)+\sum_{i=1}^{m/2}D(u_i,u_{m+1-i}).\]
\end{lemma}

Next, we will make $E_1$ and $E_2$ as small as possible by reordering the teams.

First, we consider $E_2$. The extra cost is the sum of the weight of edges $\{u_iu_{m+1-i}\}_{i=1}^{m/2}$ in $H$, which form a matching in $H$.
Our algorithm is to reorder $u_i$ such that $\{u_iu_{m+1-i}\}_{i=1}^{m/2}$ is a minimum perfect matching in $H$.
Note that $H$ is a complete graph on $m$ (even) vertices and then we can use $O(m^3)$ time
algorithm to find the minimum perfect matching $M_H$ in $H$.
Our algorithm will reorder $u_i$ such that $\{u_iu_{m+1-i}\}_{i=1}^{m/2}=M_H$. For the cost of $M_H$, we have that

\begin{eqnarray} \label{M_H}
E_2=D_{M_H}\leq \frac{1}{m-1}D_H.
\end{eqnarray}

Second, we consider $E_1$. Our idea is to choose $u_m$ such that $\sum_{i=2}^{m-2}D(u_m,u_i)$ is minimized.
Note that once $u_m$ is determined, super-team $u_1$ is also determined by the matching $M_H$ ($u_mu_1$ should be an edge in $M_H$).
After determining $u_m$ and $u_1$ together, we sill need to decider $u_{m-1}$.
We first let $u_m$ be the super-team such that $\sum_{i=2}^{m-1}D(u_m,u_i)$ is minimized (There are $m$ possible candidates for $u_m$). Thus, we have that
\[\sum_{i=2}^{m-1}D(u_m,u_i)\leq
\frac{2(D_H-D_{M_H})}{m}.
\]
Then we let $u_{m-1}$ be the super-team such that $D(u_m,u_{m-1})\geq D(u_m,u_i)$ for all $2\leq i \leq m-2$. Thus, we have that
\begin{eqnarray} \label{L}
E_1 = \sum_{i=2}^{m-2}D(u_m,u_i)\leq \sum_{i=2}^{m-1}D(u_m,u_i)\frac{m-3}{m-2}\leq \frac{2(m-3)(D_H-D_{M_H})}{m(m-2)}.
\end{eqnarray}

By (\ref{eqn_GH}), (\ref{eqn_lowerbound}), (\ref{M_H}) and (\ref{L}), we know that the total extra cost of our schedule is
\begin{eqnarray} \label{Even}
\begin{array}{*{20}l}
E_1+E_2&\leq& D_{M_H}+\frac{2(m-3)(D_H-D_{M_H})}{m(m-2)}\\
&=&(1-\frac{3}{m}+\frac{1}{m-2})D_{M_H}+(\frac{3}{m}-\frac{1}{m-2})D_H\\
&\leq&(\frac{3}{m}-\frac{3}{m(m-1)})D_H\\
&\leq&(\frac{3}{2m}-\frac{3}{2m(m-1)})LB=(\frac{3}{n}-\frac{6}{n(n-2)})LB.
\end{array}
\end{eqnarray}

Next, we analyze the running-time bound of our algorithm.
Our algorithm first uses $O(n^3)$ time to compute the minimum perfect matching $M$ and the minimum perfect matching $M_H$. It takes $O(n^2)$ time for us to determine $u_m$ and $u_{m-1}$ such that (\ref{M_H}) and (\ref{L}) hold and the remaining construction of the schedule also use $O(n^2)$ time. Thus, our algorithm runs in $O(n^3)$ time.

\begin{theorem}\label{result1}
For TTP-2 on $n$ teams where $n\geq 8$ and $n\equiv 0 \pmod 4$, a feasible schedule can be computed in $O(n^3)$ time such that the total traveling distance is at most $(1+{\frac{3}{n}}-{\frac{6}{n(n-2)}})$ times of the independent lower bound.
\end{theorem}

\subsection{The tightness of the analysis}
Next, we show that the analysis of our algorithm is tight, i.e., the approximation ratio in Theorem~\ref{result1} is the best for our algorithm. We show an example that the ratio can be reached.

In the example, the distance of each edge in the minimum perfect matching $M$ is 0 and the distance of any other edge in $G$ is 1.
We can see that the triangle inequality property still holds.
Let $E(G)$ denote the edge set of graph $G$. By (\ref{eqn_lowerbound}), we know that the independent lower bound of this instance is
\begin{eqnarray}\label{worst_case}
2D_G+nd(M)=2(\left|E(G)\right|-\left|M\right|)=n(n-2).
\end{eqnarray}

In this instance, the extra costs of a normal super-game, left super-game and last super-game are 0, 4 and 4, respectively.
In our schedule, there are $m-3$ left super-games and $\frac{m}{2}$ last super-games in total.
Thus, the total extra cost of our schedule is $4\times(m-3)+4\times\frac{m}{2}=3n-12$. Thus, the ratio is
\begin{eqnarray}\label{worst_ratio}
1+\frac{3n-12}{n(n-2)}=1+\frac{3}{n}-\frac{6}{n(n-2)}.
\end{eqnarray}

This example only shows the ratio is tight for this algorithm. However, it is still possible that some other algorithms can achieve a better ratio.

\section{Experimental Results}
To test the performance of our schedule algorithm, we will implement it on well-known benchmark instances.
The above construction method can guarantee a good approximation ratio. However, for experimentations, we may still be able to get further improvements by some heuristic methods.
For experiments, we will also use some simple heuristic methods to get further improvements.

\subsection{Heuristics based on Local Search}
Above we have introduced a method to construct a feasible schedule for TTP-2. Based on one feasible schedule,
we may be able to get many feasible schedules by just changing the ordering of the teams.
There are total $n!$ permutations of the $n$ teams, each of which may lead to a feasible schedule.
In the above analysis, we choose the permutation such that we can get a good approximation ratio.
This is just for the purpose of the analysis. We do not guarantee this permutation is optimal.
Other permutations may lead to better results on each concrete instance.
However, the number of permutations is exponential and it is not effective to check all of them.
If we check all permutations, the running-time bound will increase a factor of $n!$, which is not polynomial and not effective.
Our idea is to only consider the permutations obtained by swapping the indexes of two super-teams and by swapping the indexes of the two teams in the same super-team.
First, to check all possible swapping between two super-teams, we will have $O(m^2)$ loops, and the running-time bound will increase a factor of $m^2$.
Second, for each last super-game between two super-teams, we consider the two orders of the two teams in each super-team and then we get four cases.
We directly compute the extra cost for the four cases and select the best one. There are $m/2$ last super-games and then we only have $O(m)$ additional time.
Note that we do not apply the second swapping for normal and left super-games since this operation will not get any improvement on them (this can be seen from the proof of Lemma~\ref{extra}).


\subsection{Applications to Benchmark Sets}
Our tested benchmark comes from~\cite{trick2007challenge}, where introduces 62 instances and most of them are instances from the real world. There are 34 instances of $n$ teams with $n\geq 4$ and $n\equiv 0 \pmod 4$. Half of them are very small ($n\leq 8$) or very special (all teams are in a cycle or the distance between any two teams is 1) and they were not tested in previous papers.
So we only test the remaining 17 instances.
The results are shown in
 Table~\ref{experimentresult}, where the column `\emph{ILB Values}' indicates the independent lower bounds, `\emph{Previous Results}' lists previously known results in~\cite{DBLP:conf/mfcs/XiaoK16}, `\emph{Before Swapping}' is the results obtained by our schedule algorithm without using the local search method of swapping, `\emph{After Swapping}' shows the results after swapping, `\emph{Our Gap}' is defined to be $\frac{After Swapping~-~ILB~Values}{ILB~Values}$ and `\emph{Improvement Ratio}' is defined as $\frac{Previous~Results~-~After Swapping}{Previous~Results}$.

\begin{table}[ht]
\centering
\begin{tabular}{ m{1.4cm}<{\centering} *{1}{|m{1.2cm}<{\centering}} *{1}{m{1.3cm}<{\centering}} *{1}{m{1.4cm}<{\centering}}*{1}{m{1.4cm}<{\centering}}*{1}{m{1.2cm}<{\centering}}*{1}{m{1.9cm}<{\centering}}  }
\hline
  Data & ILB & Previous & Before & After &Our & Improvement \\
   Set & Values &  Results & Swapping & Swapping &Gap(\%) & Ratio(\%)\\
\hline
\vspace{+1mm}
  Galaxy40 & 298484 & 307469 & 306230 & 305714 & 2.42 & 0.57 \\
  Galaxy36 & 205280 & 212821 & 211382 & 210845 & 2.71 & 0.93 \\
  Galaxy32 & 139922 & 145445 & 144173 & 144050 & 2.95 & 0.96 \\
  Galaxy28 & 89242  & 93235  & 92408  & 92291  & 3.42 & 1.01 \\
  Galaxy24 & 53282  & 55883  & 55486  & 55418  & 4.01 & 0.83 \\
  Galaxy20 & 30508  & 32530  & 32082  & 32067  & 5.11 & 1.42 \\
  Galaxy16 & 17562  & 19040  & 18614  & 18599  & 5.90 & 2.32 \\
  Galaxy12 & 8374   & 9490   & 9108   & 9045   & 8.01 & 4.69 \\
  NFL32    & 1162798& 1211239& 1199619& 1198091& 3.04 & 1.09 \\
  NFL28    & 771442 & 810310 & 798208 & 798168 & 3.46 & 1.50 \\
  NFL24    & 573618 & 611441 & 598437 & 596872 & 4.05 & 2.38 \\
  NFL20    & 423958 & 456563 & 444426 & 442950 & 4.48 & 2.98 \\
  NFL16    & 294866 & 321357 & 310416 & 309580 & 4.99 & 3.66 \\
  NL16     & 334940 & 359720 & 351647 & 350727 & 4.71 & 2.50 \\
  NL12     & 132720 & 144744 & 140686 & 140686 & 6.00 & 2.80 \\
  Super12  & 551580 & 612583 & 590773 & 587387 & 6.49 & 4.11 \\
  Brazil24 & 620574 & 655235 & 643783 & 642530 & 3.54 & 1.94 \\
\end{tabular}
\caption{Experimental results on the 17 instances with $n$ teams ($n$ is divisible by 4)}
\label{experimentresult}
\end{table}

From Table~\ref{experimentresult}, we can see that our schedule algorithm can improve all the 17 instances with an average improvement of $2.10\%$. In these tested instances, the number of teams is at most 40. So our algorithm runs very fast.
On a standard laptop with a 2.30GHz Intel(R) Core(TM) i5-6200 CPU and 8 gigabytes of memory, all the 17 instances can be solved together within 0.1 seconds before applying the local search and within 8 seconds including local search.

\section{Conclusion}
In this paper, we introduce a new schedule for TTP-2 with $n\equiv 0 \pmod 4$ and prove an approximation ratio of   $(1+{\frac{3}{n}}-{\frac{6}{n(n-2)}})$, improving the previous ratio of $(1+\frac{4}{n}+\frac{4}{n(n-2)})$ in~\cite{DBLP:conf/mfcs/XiaoK16}.
The improvement looks small. However, the ratio is quite close to 1 now and further improvements become harder and harder.
Furthermore, the new construction method is simpler and more intuitive, compared with the previous method in~\cite{DBLP:conf/mfcs/XiaoK16}.
Experiments also show that the new schedule improves the results on all tested instances in the benchmark~\cite{trick2007challenge}.
In the analysis, we can see that the extra cost of our schedule is contributed by left and last super-games.
So we can decompose the analysis of the whole schedule into the analysis of left and last super-games.
To get further improvements, we only need to reduce the number of left and last super-games.


\begin{thebibliography}{10}
\providecommand{\url}[1]{\texttt{#1}}
\providecommand{\urlprefix}{URL }
\providecommand{\doi}[1]{https://doi.org/#1}

\bibitem{anagnostopoulos2006simulated}
Anagnostopoulos, A., Michel, L., Van~Hentenryck, P., Vergados, Y.: A simulated
  annealing approach to the traveling tournament problem. Journal of Scheduling
   \textbf{9}(2),  177--193 (2006)

\bibitem{bhattacharyya2016complexity}
Bhattacharyya, R.: Complexity of the unconstrained traveling tournament
  problem. Operations Research Letters  \textbf{44}(5),  649--654 (2016)

\bibitem{campbell1976minimum}
Campbell, R.T., Chen, D.: A minimum distance basketball scheduling problem.
  Management science in sports  \textbf{4},  15--26 (1976)

\bibitem{di2007composite}
Di~Gaspero, L., Schaerf, A.: A composite-neighborhood tabu search approach to
  the traveling tournament problem. Journal of Heuristics  \textbf{13}(2),
  189--207 (2007)

\bibitem{easton2001traveling}
Easton, K., Nemhauser, G., Trick, M.: The traveling tournament problem:
  description and benchmarks. In: {CP} 2001. pp. 580--584 (2001)

\bibitem{easton2003solving}
Easton, K., Nemhauser, G., Trick, M.: Solving the travelling tournament
  problem: a combined integer programming and constraint programming approach.
  In: {PATAT} 2003. pp. 100--109 (2003)

\bibitem{goerigk2014solving}
Goerigk, M., Hoshino, R., Kawarabayashi, K., Westphal, S.: Solving the
  traveling tournament problem by packing three-vertex paths. In: Twenty-Eighth
  AAAI Conference on Artificial Intelligence. pp. 2271--2277 (2014)

\bibitem{hoshino2013approximation}
Hoshino, R., Kawarabayashi, K.i.: An approximation algorithm for the bipartite
  traveling tournament problem. Mathematics of OR  \textbf{38}(4),  720--728
  (2013)

\bibitem{imahori2010approximation}
Imahori, S., Matsui, T., Miyashiro, R.: A 2.75-approximation algorithm for the
  unconstrained traveling tournament problem. AOR  \textbf{218}(1),  237--247
  (2014)

\bibitem{kendall2010scheduling}
Kendall, G., Knust, S., Ribeiro, C.C., Urrutia, S.: Scheduling in sports: An
  annotated bibliography. Computers \& Operations Research  \textbf{37}(1),
  1--19 (2010)

\bibitem{lim2006simulated}
Lim, A., Rodrigues, B., Zhang, X.: A simulated annealing and hill-climbing
  algorithm for the traveling tournament problem. European Journal of
  Operational Research  \textbf{174}(3),  1459--1478 (2006)

\bibitem{miyashiro2012approximation}
Miyashiro, R., Matsui, T., Imahori, S.: An approximation algorithm for the
  traveling tournament problem. Annals of Operations Research  \textbf{194}(1),
   317--324 (2012)

\bibitem{rasmussen2008round}
Rasmussen, R.V., Trick, M.A.: Round robin scheduling--a survey. European
  Journal of Operational Research  \textbf{188}(3),  617--636 (2008)

\bibitem{thielen2011complexity}
Thielen, C., Westphal, S.: Complexity of the traveling tournament problem.
  Theoretical Computer Science  \textbf{412}(4),  345--351 (2011)

\bibitem{thielen2012approximation}
Thielen, C., Westphal, S.: Approximation algorithms for \protect{TTP}(2).
  Mathematical Methods of Operations Research  \textbf{76}(1),  1--20 (2012)

\bibitem{trick2007challenge}
Trick, M.: Challenge traveling tournament instances (2013), online reference at
  http://mat.gsia.cmu.edu/TOURN/

\bibitem{de1988some}
de~Werra, D.: Some models of graphs for scheduling sports competitions.
  Discrete Applied Mathematics  \textbf{21}(1),  47--65 (1988)

\bibitem{westphal2014}
Westphal, S., Noparlik, K.: A 5.875-approximation for the traveling tournament
  problem. Annals of Operations Research  \textbf{218}(1),  347--360 (2014)

\bibitem{DBLP:conf/mfcs/XiaoK16}
Xiao, M., Kou, S.: An improved approximation algorithm for the traveling
  tournament problem with maximum trip length two. In: {MFCS} 2016. pp.
  89:1--89:14 (2016)

\bibitem{yamaguchi2009improved}
Yamaguchi, D., Imahori, S., Miyashiro, R., Matsui, T.: An improved
  approximation algorithm for the traveling tournament problem. Algorithmica
  \textbf{61}(4),  1077--1091 (2011)

\bibitem{DBLP:conf/ijcai/ZhaoX21}
Zhao, J., Xiao, M.: The traveling tournament problem with maximum tour length
  two: {A} practical algorithm with an improved approximation bound. In:
  {IJCAI} 2021. pp. 4206--4212 (2021)

\end{thebibliography}
\end{document}